\documentclass[letterpaper]{IEEEtran}

\usepackage{amsthm,amsmath,amssymb,color,graphicx,caption,subcaption,comment,tikz,url,latexsym} 
\usepackage{pgfplots}
\usepackage{verbatim}
\usetikzlibrary{decorations.markings}
\usepackage{cite}
\usepackage{enumerate}
\usepackage{colortbl}
\usepgfplotslibrary{units}
\usetikzlibrary{spy,backgrounds}
\usepackage{pgfplotstable}
\usepackage{flushend}
\usepackage[normalem]{ulem}
\usepackage{multirow}
\newcommand{\stkout}[1]{\ifmmode\text{\sout{\ensuremath{#1}}}\else\sout{#1}\fi}
\newtheorem{theorem}{Theorem}\theoremstyle{definition}

\newtheorem{proposition}{Proposition}
\newtheorem{definition}{Definition}

\newcommand{\Aut}{\mathsf{Aut}}
\newcommand{\PL}{\mathsf{PL}}
\newcommand{\BLTA}{\mathsf{BLTA}}
\DeclareMathOperator{\dec}{dec}
\DeclareMathOperator{\adec}{adec}
\newcommand{\fixme}[2]{\ifx&#2&{\leavevmode\color{red}#1}\else{\leavevmode\color{red}FIXME\{}#1{\leavevmode\color{red}\}}\footnote{{\leavevmode\color{red}#2}}\PackageWarning{Fixme}{#1: #2}\fi}

\definecolor{matlab1}{rgb}{0.000, 0.447, 0.741} 
\definecolor{matlab2}{rgb}{0.850, 0.325, 0.098} 
\definecolor{matlab3}{rgb}{0.929, 0.694, 0.125} 
\definecolor{matlab4}{rgb}{0.494, 0.184, 0.556} 
\definecolor{matlab5}{rgb}{0.466, 0.674, 0.188} 
\definecolor{matlab6}{rgb}{0.301, 0.745, 0.933} 
\definecolor{matlab7}{rgb}{0.635, 0.078, 0.184} 

\begin{document}
	\title{Dynamic Frozen-Function Design for Reed-Muller Codes with Automorphism-Based Decoding}
	
	\author{Samet Gelincik,~\IEEEmembership{Member,~IEEE},	Charles Pillet,~\IEEEmembership{Member,~IEEE,}		and~Pascal Giard,~\IEEEmembership{Senior Member,~IEEE}
        \thanks{Samet Gelincik is with EURECOM, Sophia-Antipolis, France (email: samet.gelincik@eurecom.com).}
		\thanks{Charles Pillet and Pascal Giard are with LaCIME, Department of Electrical Engineering, École de technologie supérieure, Montréal, QC, Canada (email: \{charles.pillet,pascal.giard\}@lacime.etsmtl.ca). This project was partially funded by the NSERC CRD grant \#494694.}
}
	\maketitle
	\begin{abstract}
		In this paper, we propose to add dynamic frozen bits to underlying polar codes with a Reed-Muller information set with the aim of maintaining the same sub-decoding structure in Automorphism Ensemble (AE) and lowering the Maximum Likelihood (ML) bound by reducing the number of minimum weight codewords. 
		We provide the dynamic freezing constraint matrix that remains identical after applying a permutation linear transformation. 
		This feature also permits to drastically reduce the memory requirements of an AE decoder with polar-like codes having dynamic frozen bits.
		We show that, under AE decoding, the proposed dynamic freezing constraints lead to a gain of up to $0.25$\,dB compared to the ML bound of the $\mathcal{R}(3,7)$ Reed-Muller code, at the cost of small increase in memory requirements.
\end{abstract}
	\IEEEpeerreviewmaketitle
	\section{Introduction}
	Polar codes \cite{ArikanPolarCodes} have been introduced as the first provably capacity-achieving linear codes for the class of symmetric binary memoryless channels under the low-complexity Successive-Cancellation (SC) decoding algorithm.
	However, capacity is achieved for unpractical code lengths and polar codes suffer from a poor Maximum Likelihood (ML) bound in the finite-length regime.
	Polar codes concatenated with a cyclic redundancy check (CRC) code, namely CRC-aided (CA) polar codes, drastically enhance  the error-correction performance under SC-List (SCL) \cite{SCL} decoding.
	Indeed, using a CRC code both enhances the distance property and approximates a genie decoder \cite{SCL,crc1,crc2}. 
	CA-polar codes have been selected among the channel coding schemes in 5G-New Radio \cite{5GHuawei}.
	
	On the other hand, recently introduced
	Polarization-adjusted-convolutional (PAC) codes \cite{ArikanPAC} perform very close to the second-order rate approximation of the binary-input additive white Gaussian noise (BI-AWGN) in short block-length regime \cite{PPV} under SCL decoding with large list sizes, i.e., $L\geq128$. 
In \cite{ICC_row_merging} and \cite{Samet}, the concept of single-argument dynamic frozen functions was introduced, where any dynamic frozen bit is dependent on only one previous information bit. 
Also, it has been shown that properly chosen dynamic frozen bits reduce the number of minimum weight codewords or increase the minimum distance. 
	Code designs with single-argument dynamic frozen bits perform as well as PAC codes with much lower decoding complexity at short block-lengths \cite{ICC_row_merging}.
	All these modified polar codes fall into the same category, namely, \emph{pre-transformed} polar codes \cite{pre_transform}, and can be defined by their dynamic freezing constraint matrix.
	With properly chosen pre-transformation matrices, pre-transformed polar codes were proven to enhance the distance spectrum of the code without reducing the minimum distance of the underlying code\cite{pre_transform}.
	Given their distance properties, Reed-Muller codes are often chosen as underlying codes of pre-transformed polar codes for short block lengths \cite{ArikanPAC,ICC_row_merging,Samet}.
	
	Automorphism Ensemble (AE) decoding \cite{geiselhart2020automorphism} is a possible alternative to list decoding for codes exhibiting symmetry \cite{ReedMullerpldec,PSMC,PSC} such as Reed-Muller codes. 
	In AE decoding, $M$ permutations are selected in the affine automorphism group of the code to generate $M$ instances of the noisy vector.
	The $M$ permuted noisy vectors are decoded in parallel and the most likely candidate codeword is selected.
	AE decoding was first proposed for Reed-Muller codes \cite{geiselhart2020automorphism} since the affine automorphism group is isomorphic to the general affine group \cite{WilliamsSloane}.
	AE decoding was proposed for polar codes in \cite{geiselhart2021automorphismPC,AE_v1,AE_v2}.
    
	Recently, affine permutations outside of the affine automorphism group of a given polar code were shown to be usable in ensemble decoding \cite{PermDecRussianSubcode}.
	Each of the permutations generates a new pre-transformed polar code.
    In \cite{geiselhart_subcode}, affine permutations are successfully applied to CA-polar codes decoded with AE-Belief Propagation. 
    The error-correction performance is improved with respect to \cite{BPL_crcaided}.
    
    In this work, we integrate the single-argument dynamic frozen function concept into automorphism-based code design. 
    In that sense, we propose a dynamic freezing constraint matrix for Reed-Muller codes that remains identical for a large group of permutations.
    The proposed dynamic freezing constraint matrix permits to lower the ML bound and achieves better error-correction performance under AE-SCL with limited list size $L$ and automorphisms used $M$.
    Moreover, the proposed dynamic freezing constraint matrix drastically reduces the memory required by AE decoding algorithm when pre-transformed polar codes are used.
    
	\section{Preliminaries}

\subsection{Reed-Muller Codes, Polar Codes and Monomial Formalism}
A polar code $(N,K)$ of length $N=2^n$ and dimension $K$ is a binary block code defined by the binary kernel $\mathbf{G}_2=\left[\begin{smallmatrix}1 &0\\ 1 &1\end{smallmatrix}\right]$, the transformation matrix $\mathbf{G}_N=\mathbf{G}_2^{\otimes n}$, where $(\cdot)^{\otimes n}$ denotes the $n^{th}$ Kronecker power, the information set $\mathcal{I}\subseteq\{0,\dots,N-1\}$ and the frozen set $\mathcal{F}=\mathcal{I}^{c}$. 
Encoding is performed as $\mathbf{x}=\mathbf{u}\cdot \mathbf{G}_N$ with the input vector $\mathbf{u}=\left(u_0,\dots,u_{N-1}\right)$ generated by assigning $u_i=0$ if $i\in\mathcal{F}$ and storing information in the $K$ bit-channels stated in $\mathcal{I}$.
In polar coding, $\mathcal{I}$  includes the $K$ most reliable bit-channels resulting from channel polarization \cite{ArikanPolarCodes}. 
A polar-like code may use a different criteria to define $\mathcal{I}$, not necessarily reliability. 
The $K$ rows of $\mathbf{G}_N$ stated in $\mathcal{I}$ form the generator matrix $\mathbf{G}$ of the polar-like code.

The relation between polar codes and Reed-Muller codes was already pointed out in \cite{ArikanPolarCodes}.
That is, the row-selection of $\mathbf{G}_N$ to generate $\mathbf{G}$ for a $\mathcal{R}(r,n)$ Reed-Muller code \cite{Reed,Muller} is based on the Hamming weight of the rows. 
Hence, Reed-Muller codes are polar-like codes.
The Hamming weight of row $k$ is $2^{i_1(\mathbf{b}_k)}$ \cite{ArikanPolarCodes,ICC_row_merging}, where $\mathbf{b}_k=(b_{k,0},\dots,b_{k,n-1})$ is the binary representation of $k$ and $i_1(\cdot)$ returns the Hamming weight.
The binary representation of indices in $\mathcal{I}_{\mathcal{R}(r,n)}$ verifies 
\begin{align}
    \mathcal{I}_{\mathcal{R}(r,n)} = \left\{k\,:\,i_1(\mathbf{b}_k)\geq n-r\right\}.
\end{align}

Note that the row $\mathbf{g}_k$ of $\mathbf{G}_N$ can be represented as the evaluation of a monomial in $\mathbb{F}_2[z_0,\ldots, z_{n-1}]$, where the coordinate indexing increases from right to left:
\begin{align}
		\mathbf{g}_k=\left\{\prod_{j\in \mathcal{P}_{0}(\mathbf{b}_k) }\hspace{-0.3cm}z_j: \mathbf{z}\in \mathbb{F}^{n}_2\right\},\label{Eq:g_rows_monmial}
	\end{align}
	i.e.,
	\begin{align*}
		g_{k,m}=\prod_{j\in \mathcal{P}_{0}(\mathbf{b}_k) }\hspace{-0.3cm}z_j, \quad [z_0,\ldots, z_{n-1}]=\mathbf{b}_{N-m-1},
	\end{align*}
where $\mathcal{P}_{0}(\cdot)$ is the set of coordinate indices with $0$ in $\mathbf{b}_k$.

Hence, any row $\mathbf{g}_k$ corresponds to a monomial obtained with respect to binary representation of $k$:
	\begin{align}\label{eq:monomial2}
	\mathsf{f}_k(\mathbf{z})= \prod_{j\in \mathcal{P}_{0}(\mathbf{b}_k)} \hspace{-0.2cm}z_j,
	\end{align}
i.e., the monomial and binary representations are alternative representations of the same entity: polar matrix rows. 
Reed-Muller codes can thus be described through monomial formalism in such that any codeword is obtained as an evaluation of a linear combination of some monomials in the monomial set $\mathcal{M}_{\mathcal{R}(r,n)}$ corresponding to a given information vector, where
	\begin{align}\label{eq:monomial3}
	    \mathcal{M}_{\mathcal{R}(r,n)}=\left\{\mathsf{f}_k(\mathbf{z}): i_1(\mathbf{b}_k)\geq n-r\right\}.
	\end{align}
	

	\subsection{Pre-transformed Polar Codes}
For a given $K$-bit information vector $\mathbf{v}$ and a pre-transformation matrix $\mathbf{W}\in\mathbb{F}_2^{K\times N}$, encoding of a pre-transformed polar code is performed as $\mathbf{x}=\mathbf{u}\mathbf{G}_N$ with  $\mathbf{u}=\mathbf{v}\mathbf{W}$. 
Pre-transformed polar codes are defined by the dynamic freezing constraint matrix $\mathbf{V}\in\mathbb{F}_2^{(N-K)\times N}$ verifying $\mathbf{W}\mathbf{V}^\top=\mathbf{0}$ \cite{polar_subcode}, where $(\cdot)^\top$ denotes the transpose of a matrix.
$\mathbf{V}$ transcribes the bits in $\mathcal{F}$, frozen at $0$ (with a single $1$ in the row) or data-dependent (several $1s$ in one row).
Properly chosen pre-transformation matrices may reduce the number of low-weight codewords, enhancing the distance properties in the finite length regime \cite{pre_transform,Samet}.
Reed-Muller codes are often adopted as underlying codes for pre-transformed polar codes at short block lengths \cite{ArikanPAC,ICC_row_merging,Samet}.
	\subsection{Automorphism Group of Reed-Muller Codes}
An automorphism $\pi$ of a code $\mathcal{C}$ is a permutation on $\{0,\dots,N-1\} $, mapping every codeword $x \in \mathcal{C}$ into another codeword $\pi(x) \in \mathcal{C}$. 
The automorphism group $\Aut(\mathcal{C})$ of a code $\mathcal{C}$ is the set containing all automorphisms of the code. 
The general affine group of order $n$, $\mathsf{GA}(n)$,  is defined as the set of all transformations of $n$ variables described by  
\begin{align}
	\mathbf{z} & \mapsto \mathbf{z'} = \mathbf{A} \mathbf{z} + \mathbf{b}\; \text{mod}\, 2  ,
	\label{eq:GA}
\end{align}
with $\mathbf{z} , \mathbf{z'} \in \mathbb{F}_2^n$,  $\mathbf{A}\in\mathbb{F}_2^{n\times n}$ an invertible matrix and $\mathbf{b}\in\mathbb{F}_2^{n\times 1}$. 
The \emph{affine automorphism group} $\mathcal{A}(\mathcal{C})\subseteq \Aut(\mathcal{C})$ of any Reed-Muller code is known to be general affine group $\mathsf{GA}(n)$ \cite{WilliamsSloane}.
The affine automorphism group of polar codes contains the lower-triangular-affine group ($\mathsf{LTA}$) of order $n$. 

The two aforementioned groups can be described by the block-lower-triangular affine (BLTA) group \cite{geiselhart2020automorphism}, namely the group of affine transformations having a block-lower-triangular transformation matrix. 
A $\BLTA$ group is recovered by the sizes of the blocks alongside the diagonal, defining the block structure $\mathbf{S}=(s_1,\ldots,s_l)$, $s_1+\ldots+s_l=n$. 
$\mathsf{LTA}$ and $\mathsf{GA}(n)$ correspond to $\BLTA(1,\dots,1)$ and $\BLTA(n)$, respectively.

The permutation linear transformation $\PL$ corresponds to the group in which $\mathbf{A}$ is a permutation matrix and $\mathbf{b}=0$. 
Note that any $\PL$ automorphism (except identity) maps each monomial $\mathsf{f}_k(\mathbf{z})$ in $\mathcal{M}_{\mathcal{R}(r,n)}$ to a different $\mathsf{f}_{k^{'}}(\mathbf{z})$ in $\mathcal{M}_{\mathcal{R}(r,n)}$ with the same degree. 
This leads to the following corollary.
\begin{proposition}\label{prop:prop1}
Let $\pi\in\PL$ be an automorphism with $\mathbf{A}$ as the corresponding permutation matrix. 
Then, $\mathsf{f}_k(\mathbf{A}\mathbf{z})=\mathsf{f}_{k^{'}}(\mathbf{z})$
such that
\begin{align}
    \mathbf{b}_{k^{'}}= \mathbf{A}\mathbf{b}_k. 
\end{align}
\end{proposition}
\begin{proof}
Application of $\pi$ to row $\mathbf{g}_k$ is given as
\begin{align*}
    \pi(\mathbf{g}_k)&\overset{(a)}{=}\pi(\{\mathsf{f}_k(\mathbf{z}), \, \mathbf{z}\in \mathbb{F}_{2}^{n}\}),\\
    & =\{\mathsf{f}_k(\mathbf{A}\mathbf{z}), \, \mathbf{z}\in \mathbb{F}_{2}^{n}\},\\
    & \overset{(b)}{=} \left\{\prod_{j\in \mathcal{P}_{0}(\mathbf{b}_k)} \hspace{-0.2cm}z_j^{'}, \, \mathbf{z}\in \mathbb{F}_{2}^{n}\right\},\\
    & \overset{(c)}{=} \left\{\prod_{j\in \mathcal{P}_{0}(\mathbf{A}\mathbf{b}_k)} \hspace{-0.2cm}z_j, \, \mathbf{z}\in \mathbb{F}_{2}^{n}\right\},\\
     & = \left\{\prod_{j\in \mathcal{P}_{0}(\mathbf{b}_{k^{'}})} \hspace{-0.2cm}z_j, \, \mathbf{z}\in \mathbb{F}_{2}^{n}\right\},\\
     & = \{\mathsf{f}_{k^{'}}(\mathbf{z}), \, \mathbf{z}\in \mathbb{F}_{2}^{n}\},    
\end{align*}
where $z_j^{'}$ is obtained from  vector multiplication of the $j^{th}$ row of $\mathbf{A}$ with ${[ z_0,\ldots, z_{n-1}]}^\top$, (a) comes from \eqref{Eq:g_rows_monmial} and~\eqref{eq:monomial2} while (b) and (c) are induced by the permutation matrix $\mathbf{A}$.
\end{proof}
\subsection{Automorphism Ensemble Decoder of Reed-Muller Codes}
Given a decoder $\dec$ for a code $\mathcal{C}$, the corresponding automorphism decoder $\adec$ is given by
\begin{equation}
\label{eq:adec}
\adec(y,\pi) = \pi^{-1}\left( \dec(\pi(y)) \right),
\end{equation}
where $y$ is the received signal and $\pi \in \mathcal{A}(\mathcal{C})$. 
An \emph{automorphism ensemble} (AE) decoder \cite{geiselhart2020automorphism} consists of $M$ $\adec$ instances running in parallel with $\pi\in\mathcal{A}(\mathcal{C})$.
Each of the $M$ $\adec$ returns a codeword candidate and the least-squares metric is used to select the most likely candidate. 

In \cite{PermDecRussianSubcode}, permutations outside of $\mathcal{A}(\mathcal{C})$ are used and proved to result in a pre-transformed polar code.
For a pre-transformed polar code defined by $\mathbf{V}$, applying a permutation $\pi\in\mathsf{GA}(n)$ results in a pre-transformed polar code defined by the transformed dynamic freezing constraint matrix \cite{geiselhart_subcode}
\begin{align}
\mathbf{V}_T=\mathbf{V}\left(\mathbf{G}_N\mathbf{T}^{-1}\mathbf{G}_N\right)^\top,\label{eq:VT}
\end{align}
with $\mathbf{T}^{-1}\in\mathbb{F}_2^{N\times N}$, the post-transformation matrix of $\pi$, defined as the permutation matrix related to $\pi$.

According to \eqref{eq:VT}, the transformed dynamic freezing constraint matrix $\mathbf{V}_T$ is permutation dependent. 
Hence, if pre-transformed polar codes are used in AE decoder, every $\adec$ decoder needs to store its related transformed dynamic freezing constraint matrix leading to large memory requirements.
	
\section{Stable dynamic freezing constraint matrix}\label{sec:stable_freezing}
In this section, given that $\mathbf{T}^{-1}$ is a post-transformation matrix of a permutation in $\BLTA(n-1,1)\cap\PL$, a dynamic freezing constraint matrix $\mathbf{V}$ is proposed such that $\mathbf{V}_T$ returns an equivalent dynamic freezing constraint.
Next, $\mathcal{R}(r,n)$ Reed-Muller codes are used having $\mathcal{A}(\mathcal{R}(r,n))=\mathsf{BLTA}(n)$.

\subsection{Structure of $\mathbf{V}$}
First, the definition of two equivalent dynamic freezing constraint matrices is given.
\begin{definition}[Equivalent dynamic freezing constraint matrices]\label{eq:equivalence}
Given a pre-transformation matrix $\mathbf{W}$, two dynamic freezing constraint matrices $\mathbf{V}_1$ and $\mathbf{V}_2$ are equivalent if $\mathbf{WV}_1^\top=\mathbf{WV}_2^\top=\mathbf{0}$. 
\end{definition}
After row reduction of $\mathbf{V}_1$ and $\mathbf{V}_2$, the same dynamic freezing constraints are obtained. It follows that a given pre-transformed polar code can be defined in various ways. 
Next, the stability of $\mathbf{V}$ under a set of permutations is defined.
\begin{definition}[Stability of $\mathbf{V}$]
A dynamic freezing constraint matrix $\mathbf{V}$ is said to be \emph{stable} under the set of permutations $\mathcal{B}$ if and only if for any $\pi\in\mathcal{B}$, the transformed dynamic freezing constraint matrix $\mathbf{V}_T$ is equivalent to $\mathbf{V}$.
\end{definition}
First, the input vector $\mathbf{u}$ stores the information in the $K$ positions stated in $\mathcal{I}$.
The proposed freezing constraint is:
\begin{align}
    	\forall i\in\mathcal{F}\begin{cases} u_i=0\text{,} & \text{if } b_{i,n-1}=0 \iff i<N/2\text{,}\\
    u_i=u_j\text{,} & \text{if } b_{i,n-1}=1 \, \text{with } \mathbf{b}_j=\mathbf{\overline{b}}_i,
	\end{cases}\label{eq:constraint}
\end{align}
where $\bar{(\cdot)}$ denotes the complement operation, i.e., for any $0\leq \ell<n$, $b_{j,\ell}=1+b_{i,\ell} \text{ mod}\, 2$.

The second line in \eqref{eq:constraint} can be sub-divided as either $j\in\mathcal{F}$ or $j\in\mathcal{I}$. In the former case, we have $u_i=u_j=0$, in the latter case, the bit $i$ is a dynamic frozen bit with $u_i=u_j\in\{0,1\}$.
Hence, in most cases, the matrix $\mathbf{V}$, that transcribes the bits in $\mathcal{F}$ as induced by the constraint \eqref{eq:constraint}, can be row reduced.
The maximum number of dynamic frozen bits $D$ is given by
{\small
\begin{align}\label{eq:D}
    D=\min\left(\left|\{i\in\mathcal{F}|b_{i,n-1}=1\}\right|,\left|\{j\in\mathcal{I}|b_{j,n-1}=0\}\right|\right).
\end{align}
}%

Next, we prove that the dynamic freezing constraint matrix $\mathbf{V}$ induced by \eqref{eq:constraint}, compliant with the concept of single-argument dynamic frozen functions \cite{ICC_row_merging}, is stable after applying any permutation $\pi\in\BLTA(n-1,1)\cap\PL$.
\begin{theorem}[Stability over $\BLTA(n-1,1)\cap\PL$]\label{theo:constraint}
The dynamic freezing constraint matrix $\mathbf{V}$ induced by \eqref{eq:constraint} is stable under any permutation $\pi\in\BLTA(n-1,1)\cap\PL$.
\end{theorem}
\begin{proof}
The permutation $\pi$ is generated with the permutation matrix $\mathbf{A}\in\mathbb{F}_2^{n\times n}$. 
Let us denote by $f$, the vector such that $\mathbf{A}_{k,f(k)}=1$ for $1\leq k \leq n$, we have $f(n)=n$.

By Proposition~\ref{prop:prop1}, application of $\PL$ automorphisms to monomials can also be represented by matrix multiplication of $\mathbf{A}$ with binary representation of the corresponding indices. 
Hence, in the following, the proof is given in terms of binary representations of row indices.

For any $i\in\mathcal{F}$ with $b_{i,n-1}=1$, the bit value is data dependent on index $j\in\mathcal{I}$ with $\mathbf{b}_j=\overline{\mathbf{b}}_i$.
The indices $\mathbf{b}_i$ and $\mathbf{b}_j$ are mapped to $\mathbf{b}_{i'}=\mathbf{A}*\mathbf{b}_i$ and $\mathbf{b}_{j'}=\mathbf{A}*\mathbf{b}_j$ with
\begin{align}
    \mathbf{b}_{i'}=\begin{bmatrix} b_{i,f(1)}\\ b_{i,f(2)} \\\dots \\ b_{i,f(n)=n}\end{bmatrix},\hspace{0.05\columnwidth} \mathbf{b}_{j'}=\begin{bmatrix} b_{j,f(1)}\\ b_{j,f(2)} \\\dots \\ b_{j,f(n)=n}\end{bmatrix}=\overline{\mathbf{b}}_{i'}.
\end{align}
Note that, since $\pi \in \PL\subset\mathcal{A}(\mathcal{R}(r,n))$, the information set $\mathcal{I}$ is preserved, i.e., $\pi(\mathcal{I})=\mathcal{I}$, and for any $(i,j)$ pair there is a corresponding $(i',j')$ pair such that $i_1(\mathbf{b}_i)=i_1(\mathbf{b}_i')$ and $i_1(\mathbf{b}_j)=i_1(\mathbf{b}_j')$. Hence, by construction of dynamic frozen set and functions, by \eqref{eq:constraint}, the dynamic frozen set and corresponding functions are preserved as well.
\end{proof}
The equivalence between $\mathbf{V}$ and $\mathbf{V}_T$ for a single permutation $\pi$ can also be verified with \eqref{eq:VT}.
As an example, the $\mathcal{R}(1,3)$ code, having frozen set $\mathcal{F}=\{0,1,2,4\}$, is constrained according to \eqref{eq:constraint}.
The matrices $\mathbf{V}$ and $\mathbf{W}$ are:
\begin{align}\label{eq:matrixV}
  	\mathbf{V}=\left[ \begin{smallmatrix}
  		1 & 0 & 0 & 0 & 0 & 0 & 0 & 0 \\
  		0 & 1 & 0 & 0 & 0 & 0 & 0 & 0 \\
  		0 & 0 & 1 & 0 & 0 & 0 & 0 & 0 \\
  		0 & 0 & 0 & 1 & 1 & 0 & 0 & 0
  	\end{smallmatrix} \right], \,\,\,\,\,\,
  	\mathbf{W}=\left[ \begin{smallmatrix}
  		0 & 0 & 0 & 1 & 1 & 0 & 0 & 0 \\
  		0 & 0 & 0 & 0 & 0 & 1 & 0 & 0 \\
  		0 & 0 & 0 & 0 & 0 & 0 & 1 & 0 \\
  		0 & 0 & 0 & 0 & 0 & 0 & 0 & 1
  	\end{smallmatrix} \right].
  \end{align}
The code is permuted with $\pi_1=(0,2,1,3,4,6,5,7)$ retrieved as $\pi_1(i)=\mathrm{bi2int}\left(\left[ \begin{smallmatrix} 0 & 1 & 0\\ 1 & 0 & 0\\ 0 & 0 & 1 \end{smallmatrix} \right]*\mathbf{b}_i\right),$ with $0\leq i<8$ and $\mathrm{bi2int}(\mathbf{b}_x)=\sum_{t=0} 2^t*b_{x,t}$, a function that converts a binary number to an integer.
Its transformed dynamic freezing constraint matrix $\mathbf{V}_{T_1}$ is shown in \eqref{eq:matrixVT_1_2}.
$\mathbf{V}_{T_1}$ and $\mathbf{V}$ are equivalent, as expected since $\pi_1\in\mathsf{BLTA}(2,1)\cap \mathsf{PL}$.
  
  Some transformed dynamic freezing constraint matrices $\mathbf{V}_T$ generated after applying permutation $\pi\notin \mathsf{BLTA}(2,1)\cap \mathsf{PL}$ are equivalent to $\mathbf{V}$ such as $\pi_2=(7,2,5,0,3,6,1,4)\in \mathsf{LTA}$. 
  $\pi_2$ is retrieved as $\pi_2(i)=\mathrm{bi2int}\left(\left[ \begin{smallmatrix} 1 & 0 & 0\\ 0 & 1 & 0\\ 1 & 0 & 1 \end{smallmatrix} \right]*\mathbf{b}_i+\left[ \begin{smallmatrix} 1\\ 1\\ 1 \end{smallmatrix} \right]\right)$.
  The transformed dynamic freezing constraint matrix $\mathbf{V}_{T_2}$ is:
  \begin{align}\label{eq:matrixVT_1_2}
  	\mathbf{V}_{T_1}=\left[ \begin{smallmatrix}
  		1 & 0 & 0 & 0 & 0 & 0 & 0 & 0 \\
  		0 & 0 & 1 & 0 & 0 & 0 & 0 & 0 \\
  		0 & 1 & 0 & 0 & 0 & 0 & 0 & 0 \\
  		0 & 0 & 0 & 1 & 1 & 0 & 0 & 0
  	\end{smallmatrix} \right], \,\,\,\,\,\,
  	\mathbf{V}_{T_2}&=
  	\left[ \begin{smallmatrix}
  	    1 & 0 & 0 & 0 & 0 & 0 & 0 & 0 \\
  		1 & 1 & 0 & 0 & 0 & 0 & 0 & 0 \\
  		1 & 0 & 1 & 0 & 0 & 0 & 0 & 0 \\
  		1 & 0 & 1 & 1 & 1 & 0 & 0 & 0
  	\end{smallmatrix} \right].
  \end{align}
  After row reduction on $\mathbf{V}_{T_2}$, $\mathbf{V}$ and $\mathbf{V}_{T_2}$ have the same dynamic freezing constraint, i.e., $\mathbf{V}$ and $\mathbf{V}_{T_2}$ are equivalent.
  We also note that
  \begin{align}
      \mathbf{WV}^\top=\mathbf{W}\mathbf{V}_{T_1}^\top=\mathbf{W}\mathbf{V}_{T_2}^\top=\mathbf{0},
  \end{align} satisfying Definition~\ref{eq:equivalence}.
  However, we conjecture that the transformed dynamic freezing constraint is not equivalent with respect to $\mathbf{V}$ for longer codes and $\pi\notin\mathsf{BLTA}(n-1,1)\cap \mathsf{PL}$.

Next, we study dynamic freezing constraint matrices that are variants of $\mathbf{V}$ as described in \eqref{eq:constraint}.
Let $\mathcal{F}_{d_i}=\{k\in\mathcal{F}:b_{k,n-1}=1, \, i_1(\mathbf{b}_k)=i\}$ and $\mathcal{I}_{d_i}=\{k\in\mathcal{I}:b_{k,n-1}=0, \, i_1(\mathbf{b}_k)=i\}$. 
By \eqref{eq:constraint}, $\mathbf{V}$ is obtained by setting  $\mathbf{u}_{\mathcal{F}_{d_i}}=\mathbf{u}_{\mathcal{I}_{d_{n-i}}}$, where applicable. 
Then, any variant of $\mathbf{V}$ is obtained by setting $\mathbf{u}_{\mathcal{F}_{d_i}}=\mathbf{{0}}$ for some $i\in \{1,\dots,n\}$.
\begin{proposition}
    The alternative dynamic freezing constraint is also stable under $\BLTA(n-1,1)\cap\PL$.
\end{proposition}
\begin{proof}
The proof of Theorem~\ref{theo:constraint} corresponds to the case having $D$ dynamic frozen bits.
Freezing to $0$ all the indices whose binary representations share the same Hamming weight does not violate the validity of Theorem~\ref{theo:constraint} since the permutation matrix $\mathbf{A}$ maps any binary number to another one with the same Hamming weight. 
\end{proof}

\noindent Consider the $\mathcal{R}(3,7)$ Reed-Muller code. 
By \eqref{eq:constraint}, each dynamic frozen bit of the sets $\mathcal{F}_{d_1},\mathcal{F}_{d_2}$ and $\mathcal{F}_{d_3}$ is associated with an information bit of the sets $\mathcal{I}_{d_6},\mathcal{I}_{d_5}$ and $\mathcal{I}_{d_4}$, respectively.
Then, designing a pre-transformation matrix with respect to any combination of dynamic frozen constraints $\mathcal{F}_{d_1}$, $\mathcal{F}_{d_2}$ and $\mathcal{F}_{d_3}$ satisfies the $\PL$ automorphism imposed by Theorem~\ref{theo:constraint}.

Proposition~\ref{prop:total_combi} gives the total number of stable dynamic freezing constraint matrices under $\BLTA(n-1,1)\cap\PL$.
\begin{proposition}\label{prop:total_combi}
    The $\mathcal{R}(r,n)$ Reed-Muller code has $\min\left(2^{n-r-1},2^{r}\right)$ stable dynamic freezing constraint matrices under $\BLTA(n-1,1)\cap\PL$.
\end{proposition}
\begin{proof}
The number of stable dynamic freezing constraint matrices corresponds to the number of existing combinations of dynamic freezing constraints $\mathcal{F}_{d_i}$ with $\mathcal{I}_{d_{n-i}}$.
The size of sets $\mathcal{F}_{d_i}$ and $\mathcal{I}_{d_{n-i}}$ is $r$-dependent.
For a $\mathcal{R}(r,n)$ Reed-Muller code, the number of non-empty sets $\mathcal{F}_{d_i}$ and $\mathcal{I}_{d_i}$ is respectively $n-r-1$ and $r$.
Hence, the number of stable dynamic freezing constraint matrices corresponds to the number of $k$-combinations for all $0\leq k \leq \min(n-r-1,r)$, i.e., $\min\left(2^{n-r-1},2^{r}\right)$.
The dynamic freezing constraint matrix having zero dynamic frozen bit is included in the formula.
\end{proof}
For the $\mathcal{R}(3,7)$ Reed-Muller code, there are $\min\left(2^{3},2^{3}\right)=8$ stable dynamic freezing constraint matrices.
Each  combination returns a different number of low-weight codewords affecting the ML bound of the resulting pre-transformed code.
However, the best combination is not in the scope of this paper.

\subsection{Memory Requirements}
\begin{table}[t]
	\centering
	\caption{Memory required by AE under three scenarios} 
	\label{tab:mem_code}
	{\begin{tabular}{|c|c|c|c|c|c|}
		\cline{1-6}
		$\mathbf{V}_T$  & Decoder & \multicolumn{4}{c|}{$\mathcal{R}(r,n)$}\\
		\cline{3-6}		
		generated & knowledge & (3,7) & (3,8) & (4,8) & (5,8)\\
		\hline
		\hline
        Stable & Irrelevant &22&29&29&8\\
		\hline
		Not& \multicolumn{1}{r|}{Known $\pi$} & 1889 & 5724 & 4199 & 1140\\
		\cline{2-6}
		stable&\multicolumn{1}{r|}{Unknown $\pi$}&\cellcolor{gray!30}65536 & \cellcolor{gray!30}333824 & \cellcolor{gray!30}190464 & \cellcolor{gray!30}75776\\
		\hline
		\multicolumn{2}{c}{} & \multicolumn{4}{|c|}{Additional memory required (in bits)}\\
		\cline{3-6}
	\end{tabular}}
\end{table}
In the following, the additional memory requirements to perform AE decoding with pre-transformed polar codes are evaluated.
The freezing constraint evaluated is that of \eqref{eq:constraint} with $M=8$ permutations.
Any of the proposed stable dynamic freezing constraint matrix can be synthesized with a maximum of $D$ bits (see \eqref{eq:D}). 
For any code rate, $D$ has $\frac{N}{2}$ as the upper bound.
In the following, the dynamic and fixed decoder scenarios are studied.  
In the dynamic case, the permutations are unknown. 
Each $\adec$ needs to store their corresponding transformed dynamic freezing constraint matrix $\mathbf{V}_T$ since $\mathbf{V}_T$ is permutation dependent \eqref{eq:VT}.
In the fixed scenario, the $M$ permutations are known \emph{a priori} and a row reduction of each $\mathbf{V}_T$ is performed, reducing the amount of memory required.

Table~\ref{tab:mem_code} shows the memory required, in bits, to store the dynamic freezing constraint $\mathbf{V}$ for AE with $M=8$ permutations for the aforementioned cases.
If permutations in $\BLTA(n-1,1)\cap\PL$ are used, regardless of the scenario, $\mathbf{V}$ is stable and shared among the $M=8$ $\adec$ instances. 
As a result, at the cost of a reduction in permutation diversity, storing $\mathbf{V}$ requires less memory, i.e., $D \ll N(N-K)$ bits. 
Table~\ref{tab:mem_code} shows that $D\leq29$ for Reed-Muller codes of length $N\leq256$.

The last two rows of Table~\ref{tab:mem_code} show the memory required by the fixed and dynamic scenarios, respectively, for permutations in $\mathsf{LTA}\not\subset\BLTA(n-1,1)\cap\PL$.
The matrix $\mathbf{V}$ is thus not stable and cannot be shared among the $M=8$ $\adec$ instances.
The dynamic scenario requires to store the $M$ transformed dynamic constraints; $MN(N-K)$ bits are required in total.
With regard to memory requirements, the dynamic scenario is an upper bound of the fixed scenario (highlighted in gray in Table~\ref{tab:mem_code}).
For the fixed scenario, the permutations are known in advance as well as the $8$ transformed matrices $\mathbf{V}_T$. 
The $\mathbf{V}_T$ matrices can thus be row-reduced before storage, reducing the memory requirements for the fixed case by over 97\%, as can be seen by comparing the last two rows of Table\,\ref{tab:mem_code}.

For example, the proposed constraint reduces the memory requirements by 98.8\% for the $\mathcal{R}(3,7)$ Reed-Muller code with respect to the fixed scenario having permutations in $\mathsf{LTA}$.
\section{Simulation results}\label{sec:sim_results}
\begin{table}[t]
	\centering
	\caption{Number of codewords with weight $w$ for $\mathcal{R}(3,7)$-based codes and various dynamic freezing constraint matrices} 
	\label{tab:weight}
	\begin{tabular}{c|c|c|c|c|c|}
		\cline{2-6}
		 & \multicolumn{5}{c|}{$w$}\\
		\cline{2-6}	
		 & 16 & 18 & 20 &22 & 24\\
		
		\hline
		\multicolumn{1}{|c|}{\cellcolor{matlab1}$\mathbf{V}_0$}&$94488$&0&0&0&74078592\\
		\hline
		\multicolumn{1}{|c|}{\cellcolor{matlab5}$\mathbf{V}_d$}&$20760$&0&$203420$&$>0$&$>0$\\
		\hline
		\multicolumn{1}{|c|}{\cellcolor{matlab2}$\mathbf{V}_D$}&$28632$&$13504$&$172800$&$>0$&$>0$\\
		\hline
	\end{tabular}
\end{table}
In this section, simulation results with $\mathcal{R}(3,7)$ as the underlying code are performed over the AWGN channel with BPSK modulation.
Three different stable dynamic freezing constraint matrices are compared: $\mathbf{V}_0$ having no dynamic frozen bits, $\mathbf{V}_D$ with all $D=22$ dynamic frozen bits, and $\mathbf{V}_d$ with $0< d=15<D$ dynamic frozen bits. 
The $d$ dynamic frozen bits correspond to all indices $k\in\mathcal{F}$ with $i_1(\mathbf{b}_k)=3$.

The codes are decoded under AE decoding with $M=8$ automorphisms in $\BLTA(n-1,1)\cap\PL$ and having SCL with $L=16$ as sub-decoders.
This decoder is denoted as AE-8-SCL-16 in the following.
The code defined by $\mathbf{V}_0$ is the $\mathcal{R}(3,7)$ Reed-Muller code and AE decoding was shown to perform close to the ML bound \cite{geiselhart2020automorphism,ML_results}. 
By lowering the ML bound with the proposed freezing constraint, a gain is expected at the cost of a small memory overhead.

Figure~\ref{fig:sim} shows the error-correction performance of the three codes of interest under AE-8-SCL-16 decoding and their ML bound approximated with the truncated union bound (computed from the number of low-weight codewords). 
For reference, Figure~\ref{fig:sim} also shows 2 other pre-transformed polar codes: a PAC code \cite{ArikanPAC} and a 5G CA-polar code \cite{5GHuawei}.
Both are decoded with SCL-128, thus having similar complexity but higher latency with respect to AE-8-SCL-16.

Table~\ref{tab:weight} shows the number of low-weight codewords computed as in \cite{AdaptiveSCL} for codes defined by $\mathbf{V}_D,\, \mathbf{V}_d$ and $\mathbf{V}_0$. 
It can be seen that the $\mathcal{R}(3,7)$ Reed-Muller code has a large number of minimum weight ($w=16$) codewords with respect to the two other pre-transformed polar codes investigated.
The pre-transformed polar code defined by $\mathbf{V}_D$ has $28632$ codewords of weight $w=16$, a reduction of $70$\% with respect to $\mathcal{R}(3,7)$, but has generated $13504$ codewords of weight $w=18$. 
From Table~\ref{tab:weight}, the results for $\mathbf{V}_d$ show that freezing some of the dynamic frozen bits of $\mathbf{V}_D$ to $0$ allows to reduce the amount of codewords of weight $w=16$ by approximately 27\%, i.e., from $28632$ down to $20760$ (78\% with respect to $\mathcal{R}(3,7)$) while having $0$ codeword of weight $w=18$.
Comparing the union-bound curve for $\mathbf{V}_d$ against either of the other two bound curves in Figure~\ref{fig:sim}, it can be seen that a reduction in the number of low-weigth codewords lowers the ML bound.
\begin{figure}[t]
\centering
\begin{tikzpicture}

  \begin{semilogyaxis}[%
    width=\columnwidth,
    height=0.7\columnwidth,
    xmin=1.5, xmax=4,
    xtick={1.5,2,...,4},
    xlabel={$E_b/N_0,\,\mathrm{dB}$},
    xlabel style={yshift=0.4em},
    ymin=1e-6, ymax=1e-1,
    ylabel style={yshift=-0.1em},
    ylabel={BLER},
    yminorticks, xmajorgrids,
    ymajorgrids, yminorgrids,
    legend style={at={(0.00,0.00)},anchor=south west},
    legend style={legend columns=1, font=\scriptsize, row sep=-0.5mm},
    legend style={fill=white, fill opacity=0.8, draw opacity=1,text opacity=1}, 
    legend style={inner xsep=0pt, inner ysep=-1pt}, 
    legend cell align={left}, 
    mark size=1.6pt, mark options=solid,
    ]
        
    \addplot[dashed,color=black, mark=triangle, line width=0.8pt, mark size=2.8pt]
    table[row sep=crcr]{%
    0 0.5935\\
    0.5 0.3436 \\
    1 0.2079 \\
    1.5 0.0844 \\
    2 0.0250 \\
    2.5 0.0049 \\
    3 4.8710e-04 \\
    3.5 5.3388e-05\\
    4  7.0000e-06\\
    };
    \addlegendentry{SCL-128 - 5G CA-PC}
    
    \addplot[dashed,color=black, mark=star, line width=0.8pt, mark size=2.8pt]
    table[row sep=crcr]{%
    0 0.4630 \\
    0.5 0.2481  \\
    1 0.1164  \\
    1.5 0.0433  \\
    2 0.0097  \\
    2.5 0.0014  \\
    3 2.0190e-04  \\
    3.5 1.1200e-05 \\
    4  6.0000e-07\\
    };
    \addlegendentry{SCL-128 - PAC code}
    \addplot[color=matlab1, mark=square, line width=0.8pt, mark size=2.8pt]
    table[row sep=crcr]{%
    0 0.5025 \\
    0.5 0.3205  \\
    1 0.1538 \\
    1.5 0.0559 \\
    2 0.0154  \\
    2.5 0.0040 \\
    3 6.4890e-04 \\
    3.5 9.7000e-05\\
    };
    \addlegendentry{ML bound - $\mathbf{V}_0$ \cite{ML_results}}
    
    \addplot[color=matlab1, mark=10-pointed star, line width=0.8pt, mark size=2.8pt]
    table[row sep=crcr]{%
    0 0.4651  \\
    0.5 0.3003   \\
    1 0.1271  \\
    1.5 0.0513 \\
    2 0.0213  \\
    2.5 0.0051  \\
    3 7.7153e-04  \\
    3.5 1.0273e-04 \\
    4 1.1600e-05\\
    };
    \addlegendentry{AE-8-SCL-16 - $\mathbf{V}_0$}
    
    \addplot[color=matlab2, mark=diamond, line width=0.8pt, mark size=2.8pt]
    table[row sep=crcr]{%
    2 0.0074  \\
    2.5  0.0015  \\
    3 2.4342e-04  \\
    3.5 3.2462e-05 \\
    4 3.4177e-06\\
    };
    \addlegendentry{Union bound - $\mathbf{V}_D$}
    
    \addplot[color=matlab2, mark=x, line width=0.8pt, mark size=2.8pt]
    table[row sep=crcr]{%
    0 0.5236   \\
    0.5 0.3135    \\
    1 0.1590   \\
    1.5 0.0521 \\
    2  0.0241 \\
    2.5  0.0039   \\
    3 4.5465e-04   \\
    3.5 5.9671e-05  \\
    4 4.8000e-06\\
    };
    \addlegendentry{AE-8-SCL-16 - $\mathbf{V}_D$}
    
    \addplot[color=matlab5, mark=o, line width=0.8pt, mark size=2.8pt]
    table[row sep=crcr]{%
    2 0.0059  \\
    2.5  0.0011  \\
    3 1.8069e-04  \\
    3.5 2.3654e-05 \\
    4 2.4646e-06\\
    };
    \addlegendentry{Union bound - $\mathbf{V}_d$}
    
    \addplot[color=matlab5, mark=+, line width=0.8pt, mark size=2.8pt]
    table[row sep=crcr]{%
    0 0.5000\\
    0.5  0.2747\\
    1  0.1433\\
    1.5  0.0499 \\
    2   0.0165 \\
    2.5  0.0032  \\
    3 5.2057e-04  \\
    3.5 3.9610e-05   \\
    4 3.8000e-06\\
    };
    \addlegendentry{AE-8-SCL-16 - $\mathbf{V}_d$}
   \addplot[dashed,color=matlab5, mark=+, line width=0.8pt, mark size=2.8pt]
    table[row sep=crcr]{%
    0 0.5714\\
    0.5   0.3636\\
    1   0.2304\\
    1.5   0.1016\\
    2    0.0333\\
    2.5   0.0080\\
    3 0.0016\\
    3.5  1.7857e-04\\
    4  1.6000e-05\\
    };
    \addlegendentry{SCL-16 - $\mathbf{V}_d$}
  \end{semilogyaxis}

\end{tikzpicture}%

    
\caption{Decoding performance of $\mathcal{R}(3, 7)$-based codes under AE-8-SCL-16 and SCL-16 against other codes and bounds.}
\label{fig:sim}
\end{figure}
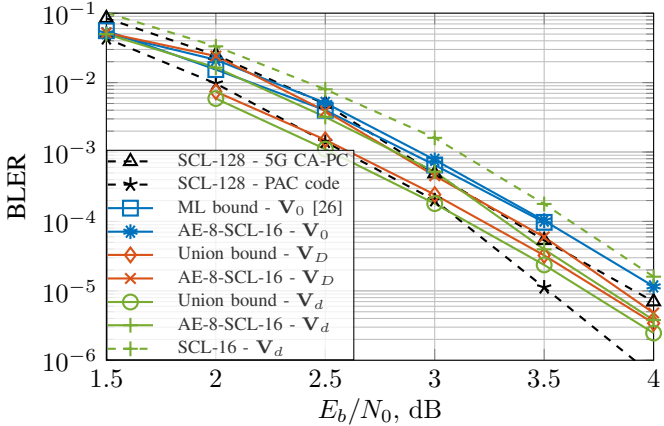

Figure~\ref{fig:sim} also shows that AE-8-SCL-16 decoding reaches the ML bound of $\mathcal{R}(3,7)$ while AE decoding with the proposed dynamic freezing constraints, i.e., $\mathbf{V}_D$ and $\mathbf{V}_d$, gains up to $0.25$\,dB with respect to this bound.
At the cost of a higher latency than AE-8-SCL-16, the 5G CA-polar code manages a similar error-correction performance under SCL-128.
Also under SCL-128, PAC codes exhibit the best performance, however the convolution inherent to PAC codes incurs more computations than our proposed constraint \eqref{eq:constraint} that only requires to read former decision values. 
The use of $8$ automorphisms improves the error-correction performance of SCL-16 by $0.3$\,dB for the code defined by $\mathbf{V}_d$. 
\section{Conclusions}
\noindent In this letter, we proposed several dynamic freezing constraint matrices for Reed-Muller codes that lower the ML bound and that are stable under permutation linear transformations.
The stability of the dynamic freezing constraints reduces the memory requirements by 99\% with respect to a dynamic AE decoder.
The decoding performance under AE-8-SCL-16 was shown to be $0.3$\,dB better than under SCL-16 while the ML bound of the $\mathcal{R}(3,7)$ Reed-Muller code is beaten by $0.25$\,dB.
\vspace{-6mm}
\bibliographystyle{IEEEbib}
\bibliography{IEEEabrv,ConfAbrv,references}
\end{document}